\documentclass{article}

\usepackage{arxiv}

\usepackage{tabularx}

\usepackage{amsmath,amssymb}
\usepackage{graphicx}
\usepackage{wrapfig}
\usepackage{array}
\usepackage[algoruled,linesnumbered,algo2e,vlined]{algorithm2e}
\usepackage{amssymb}
\usepackage{url}
\usepackage{multirow}
\usepackage{cite}

\usepackage{algorithm}
\usepackage{algorithmic}

\usepackage{comment}

 \newtheorem{theorem}{Theorem}

 \newenvironment{proof}{\trivlist \item[\hskip \labelsep{\bf proof.\/}]}{\hspace{\fill}$\Box$\endtrivlist }

\title{Outsourced Privacy-Preserving Feature Selection Based on Fully Homomorphic Encryption}

\author{
Koki Wakiyama
\And
Tomohiro I
\And
Hiroshi Sakamoto\\ \\
Kyushu Institute of Technology, 680-4 Kawazu, Iizuka, Fukuoka 820-8502, Japan \\
\texttt{wakiyama.koki809@mail.kyutech.jp}\\
  \texttt{ \{tomohiro, hiroshi\}@ai.kyutech.ac.jp} 
}

\begin{document}
\maketitle
\begin{abstract}
Feature selection is a technique that extracts a meaningful subset from a set of features in training data. When the training data is large-scale, appropriate feature selection enables the removal of redundant features, which can improve generalization performance, accelerate the training process, and enhance the interpretability of the model. This study proposes a privacy-preserving computation model for feature selection. Generally, when the data owner and analyst are the same, there is no need to conceal the private information. However, when they are different parties or when multiple owners exist, an appropriate privacy-preserving framework is required. Although various private feature selection algorithms, they all require two or more computing parties and do not guarantee security in environments where no external party can be fully trusted. To address this issue, we propose the first outsourcing algorithm for feature selection using fully homomorphic encryption. Compared to a prior two-party algorithm, our result improves the time and space complexity $O(kn^2)$ to $O(kn\log^3n)$ and $O(kn)$, where $k$ and $n$ denote the number of features and data samples, respectively. We also implemented the proposed algorithm and conducted comparative experiments with the naive one. The experimental result shows the efficiency of our method even with small datasets. 
\end{abstract}

\section{Introduction}

Feature selection is the process of identifying and selecting useful features (variables or attributes) from a dataset with the aim of improving the performance of predictive models and optimizing learning efficiency~\cite{Chandrashekar2014}. Consider a scenario in which a target variable is to be predicted using a large number of potential explanatory variables, with the goal of building a simple yet effective predictive model from a vast set of observations. For instance, in semiconductor manufacturing, more than tens of thousands of external factors can influence the production process from silicon material to the final product. If it is possible to identify in advance which of these factors are critical to product quality, then it is sufficient to construct a classification model based only on those selected features---offering a range of benefits. In practice, such preprocessing can lead to improvements in model accuracy, reductions in training time, and enhanced interpretability of the model.
Many feature selection algorithms based on various criteria have been proposed to date~\cite{Pawlak1991,Liu1998,Zhao2007,Dash2003,Arauzo-Azofra2008,Shin2011,Shin2017}.

In this study, we consider the feature selection in situations where the training data contains privacy-sensitive information. When the data owner who needs a reasonable learning model and the analyst required to construct the model are the same, privacy is generally not a concern as long as there is no data leakage. However, our focus is on privacy-preserving computation in cases where these two parties are separate. Existing frameworks to address this issue include techniques such as {\em secret sharing}~\cite{Agarwal2024} and {\em differential privacy}~\cite{Dick2023}. In this work, from the perspective of outsourcing, we focus on an approach based on {\em fully homomorphic encryption} (FHE)~\cite{Gentry2009}, which is a form of public-key cryptosystem.

In secret sharing schemes, data is split and distributed among parties, which introduces a risk of data leakage if those parties colludes maliciously. Differential privacy, on the other hand, prevents leakage by injecting random noise into the data. However, there exists a trade-off between the amount of noise added and the quality of the learning results, and it cannot guarantee complete protection against data leakage.

By contrast, FHE enables all arithmetic and comparison operations to be performed directly on encrypted data. Since only the data owner possessing the secret key can decrypt the final result, the privacy guarantees are as strong as the underlying public-key encryption. The primary limitations of FHE have been computational inefficiency and the challenge of outsourcing computation in scenarios involving multiple data owners. However, recent advancements such as TFHE~\cite{Chillotti2020} and multi-key homomorphic encryption schemes~\cite{Chen2019} are steadily overcoming these obstacles.

Therefore, we consider a simplified privacy-preserving problem for feature selection: The data owner retains the training data along with a pair of generated public and private keys. Then, the analyst receives the encrypted training data and performs feature selection using FHE scheme. The result of the computation is then returned to the data owner, who decrypts it using the private key.

Note that the proposed method constitutes a fully outsourced 
computation~\cite{Liu2015,Qiu2022} with only a single round of communication between the data owner and the analyst. This aspect significantly strengthens privacy protection. In non-fully outsourced scenarios, multiple interactions or queries from a party may be required. During such interactions, partial decryption of maliciously crafted messages from untrusted analyst could lead to unintentional leakage of the data owner's private information.

To prevent this, conventional approaches typically rely on assumptions that constrain the behavior of the analysts. For example, the {\em semi-honest} model assumes that the analyst follows the protocol faithfully but may attempt to extract as much information as possible from the data he receives. However, imposing such assumptions compromises practical applicability. In contrast, our algorithm requires only single round communication without decryption and ensures that the risk of data leakage depends solely on the strength of the cryptosystem. 
The TFHE~\cite{Chillotti2020} used in this study is proven to be at least secure against chosen-plaintext attacks (IND-CPA secure).
Consequently, it offers a significant advantage in terms of both security and practical usability compared to existing methods.
\section{Related Works and Our Contributions}
Feature selection is broadly categorized into filter method~\cite{Saeys2007}, wrapper method~\cite{Guyon2003}, and embedded method~\cite{Breiman2001}. Among these, the filter method evaluates the importance of features based on information theory, making it faster than other methods and highly generalizable due to its independence from specific models. While evaluation metrics such as mutual information are commonly used in filter methods~\cite{Peng2005}, this study focuses on feature selection based on a recently proposed {\em consistency measure}~\cite{Liu1996}. Feature selection algorithms using consistency measures have been shown to achieve both high computational efficiency and high classification accuracy on large-scale real-world datasets~\cite{Shin2011,Shin2017}. In contrast, feature selection algorithms based on wrapper or embedded methods are generally impractical for secure computation due to the high computational cost associated with searching for optimal feature subsets or training models. Therefore, this study proposes a consistency measure-based secure feature selection.

Homomorphic encryption (HE) is a framework for secure computation that leverages the homomorphic properties of public-key cryptosystems. Specifically, if encrypted integers can be added without decryption, the cryptosystem is said to be \textit{additive}. Furthermore, if it supports both operations, it is called FHE, e.g.,
RSA~\cite{RSA1978}, the first public-key cryptosystem, is multiplicative but not additive. Additionally, RSA is not a {\em probabilistic} encryption scheme---i.e., it does not produce different ciphertexts for the same plaintext in each execution---making it vulnerable to chosen-plaintext attacks.
The first probabilistic HE~\cite{Goldwasser1984} was capable of computing the sum (i.e., bitwise XOR) of encrypted bits. Subsequently, a scheme enabling the addition of integers was proposed~\cite{Paillier1999}. Later, the first HE capable of both addition and multiplication was introduced~\cite{Boneh2005}, but it allowed only single multiplication operation in overall computation, limiting its practicality for secure computation.
Eventually, the first FHE, which imposes no restrictions on the number of operations, was proposed~\cite{Gentry2009}, making arbitrary secure computations theoretically possible. In practice, fast FHE libraries have become available~\cite{TFHE}, expanding its applicability. The implementation of the algorithm proposed in this study also utilizes such a library.

Although several private feature selection methods have been proposed to date~\cite{Rao2019,Li2021,Akavia2023,Wang2025}, only a few have addressed privacy-preserving computation for the consistency measure~\cite{Ono2022}, and this study is the first to attempt a fully outsourced computation for it. This is because the computation of the optimal consistency measure reduces to the {\em set cover problem} known to be NP-hard~\cite{Garey1979}. Consequently, approximation algorithms and acceleration techniques for solving this optimization problem have been proposed~\cite{Shin2011,Shin2017}, and in this study, we realize a secure fully outsourced computation of such an algorithm.

A core technique in this approximate method is preprocessing via sorting. By sorting the entire dataset with feature values as keys, it becomes possible to efficiently narrow down the important features, enabling a heuristic approach to approximate optimal feature selection. However, executing this algorithm over encrypted data is not straightforward. In general, 
the lower bound for comparison-based sorting algorithms is $\Omega(n \log n)$, but no secure outsourced computation method is currently known that achieves this lower bound. At present, the fastest known method uses sorting networks with a time complexity of $O(n \log^2 n)$~\cite{Batcher1968}, which is also adopted in this study as preprocessing.
It should be noted, however, that if secret sharing among two or more parties is used, faster privacy-preserving sorting becomes feasible~\cite{Hamada2014}. Moreover, disregarding practical constraints, it is theoretically possible to construct sorting networks with $O(n \log n)$ time complexity~\cite{Ajtai1983}. Nonetheless, both of these approaches are outside the scope of this study.

Table~\ref{tab1} below summarizes HE-based private feature selection algorithms. Here, $k$ and $n$ denote the number of features and data samples, respectively.
A filter method~\cite{Rao2019} adopts the $\chi^2$ statistic as its evaluation criterion and employs additively homomorphic encryption~\cite{Paillier1999}. While this method has a low offline computational cost, its communication overhead is greater than that of the proposed method. Moreover, since it assumes two-party computation, it does not support fully outsourced computation. Additionally, it has not been implemented. Another two-party protocol~\cite{Ono2022} uses the same consistency measure as the proposed method and adopts the fastest FHE~\cite{Chillotti2020}. Although it also achieves exactly single round complexity as our the method, it assumes a two-party setting and therefore does not support fully outsourced computation.

\begin{table}[H]
\caption{Comparison of HE-based feature selection algorithms.\label{tab1}}
\begin{tabularx}{\textwidth}{cccccc}
\hline
\;\;\;\;\;\;		\textbf{Algorithm}\;\;\;	& \;\;\;	\;\;\;	 \textbf{Metric}\;\;\;	& \;\;\;	\;\;\;	 \textbf{Time}\;\;\; & \;\;\;	\;\;\;	\textbf{Space}\;\;\; & \;\;\;	\;\;\;	\textbf{\# Round}\;\;\; & \;\;\;	\;\;\;	 \textbf{Outsourced} \;\;\; \\
\hline
Rao et al.~\cite{Rao2019} & $\chi^2$ & $O(kn)\cdot {\cal C}_{P}$ & $O(kn)$ & $O(1)$ & partially \\
Ono et al.~\cite{Ono2022} & consistency & $O(kn^2)\cdot {\cal C}_{C}$ & $O(kn^2)$ & $1$ & partially \\
Proposed & consistency & $O(kn \log^3n)\cdot {\cal C}_{C}$ & $O(kn)$ & $1$ & fully \\
\hline
\end{tabularx}
\noindent{\footnotesize{${\cal C}_{P}$ and ${\cal C}_{C}$ are costs per single operation depending on the respective cryptosystems~\cite{Paillier1999,Chillotti2020}.  To enable a fair comparison under the same conditions, the computational complexity of Ono et al.~\cite{Ono2022} is described only with respect to feature selection task, omitting the preprocessing time required for data formatting.
The term "partially" refers to decrypting the data in part before obtaining the final result.}}
\end{table}
\section{Preliminaries}
\subsection{Consistency-Based Feature Selection}
Let $D=\{d_1,d_2,\ldots,d_n\}$ be a set of indices of data, associated with a set $F=\{f_1,f_2,\ldots,f_k\}$ of features and a class variable $C$, 
where the feature value $d(f_i)$ and the class label $d(C)$ are defined for each data $d\in D$.
Table~\ref{feature} shows an example of the triple $(D,F,C)$.

The feature selection is to find a minimal subset $F'\subseteq F$ relevant to $C$,
where $F'$ is said to be {\em consistent} if, for any $d,d'\in D$, $d(f_i)=d'(f_i)$ for all $f_i\in F'$ implies $d(C)=d'(C)$,
and $F'$ is minimal, if any proper subset of $F'$ is no longer consistent.

For example, consider the mutual information $I(F';C)$ for evaluating the relevance of $F'$ showing in Table~\ref{feature}.
We can see that $f_1$ is more important than $f_5$ due to the fact $I(f_1;C)>I(f_5;C)$.
$f_1$ and $f_2$ of Table~\ref{feature}  will be chosen to explain $C$ based on the score of $I$.
However, a closer examination of $D$ reveals that $f_1$ and $f_2$ cannot uniquely determine $C$.
In fact, we find $d_2$ and $d_5$ with $d_2(f_1)=d_5(f_1)$ and $d_2(f_2)=d_5(f_2)$ but $d_2(C)\neq d_5(C)$.
On the other hand, we can see that $f_4$ and $f_5$ uniquely determine $C$ using the formula 
$C=f_4\oplus f_5$ while $I(f_4;C)= I(f_5;C) =0$.
It becomes clear that $I(F';C)$ alone may not always lead to appropriate feature selection.

\begin{table}[H]
\caption{
An example triple $(D,F,C)$ shown {in}~\cite{Shin2017}.
} 
\label{feature}
\newcolumntype{C}{>{\centering\arraybackslash}X}
\begin{tabularx}{\textwidth}{CCCCCCC}
\hline
\boldmath{$D$} &\boldmath{$f_1$} & \boldmath{$f_2$} & \boldmath{$f_3$} &\boldmath{$f_4$} & \boldmath{$f_5$} & \;\boldmath{$C$}\; \\ \hline
$d_1$& 1 & 0 & 1 & 1 &  1 & 0 \\
$d_2$& 1 & 1 &0 &0 &0& 0 \\
$d_3$& 0 &0& 0& 1& 1& 0 \\
$d_4$& 1& 0 &1& 0& 0& 0 \\
$d_5$& 1& 1&1& 1 &0& 1 \\
$d_6$& 0 &1 &0& 1& 0 &1 \\
$d_7$& 0& 1& 0& 0 &1 &1 \\
$d_8$& 0 &0 &0& 0& 1& 1 \\ \hline
\;$I(f_i;C)$\;& \;0.189\; & \;0.189\;& \;0.049\;& \;0.000\;& \;0.000\; &  \\ 
\hline
\end{tabularx}
\end{table}

Let us review the notion of the consistency measure employed in this study.
A consistency measure $\mu:2^{F}\to[0,\infty)$ is a function to 
represent how far the data deviate from a consistent state,
where $F$ is required to satisfy {\em determinisity} (i.e., $\mu(F)=0$ iif $F$ is consistent) 
and {\em monotonicity} (i.e., $F\subseteq G$ implies $\mu(F)\geq \mu(G)$).
In this study, we focus on the binary consistency: $\mu_{\rm bin}(F) = 0$, $F$ is consistent; 1, 
otherwise~\cite{Shin2011}.
For other consistency measures, see e.g., rough set~\cite{Pawlak1991}, ICR~\cite{Dash2003}, 
and inconsistent pair~\cite{Arauzo-Azofra2008}.

As Table~\ref{feature} illustrates, the importance of a feature $f_i$ is not determined solely by itself, 
but is influenced by the relative relationship among other features. CWC~\cite{Shin2017} achieves superior feature selection compared to other statistical measures by computing the relative importance of features. Algorithm~\ref{cwc} outlines the procedure of CWC, where the consistency of the candidate set excluding $f_i$ is evaluated to determine whether to select $f_i$. Here, the order in which features $f_i $ are selected significantly affects the result, so it is important to preprocess the triple $(D, F, C)$ in advance. For example, there is a known method that determines the order of $f_i$ by sorting the triple with the values of $f_i$ as keys. This study also adopts that method.

\begin{algorithm}[H]   
\caption{CWC~\cite{Shin2011} over plaintexts~\label{cwc}}                         
                         
\begin{algorithmic}[1] 
\REQUIRE A dataset $(D,F,C)$;
\ENSURE A minimal consistent subset of $F$;
\FOR{$i=1,\ldots,k$} 
\IF{$F\setminus \{f_i\}$\text{ is consistent}}
\STATE update $F\leftarrow F\setminus \{f_i\}$;
\ENDIF
\ENDFOR
\RETURN $F$;
\end{algorithmic}
\end{algorithm}

\subsection{Computation on FHE}
The proposed FHE-based private feature selection is implemented using TFHE~\cite{TFHE}, 
one of the fastest FHE library for bitwise addition (XOR `$\oplus$')
and bitwise multiplication (AND `$\cdot$') over ciphertext.
On TFHE, any $\ell$-bit integer $m=(m_1,m_2,\ldots,m_\ell)$ is encrypted bitwise:
$E[m] \equiv (E[m_1],E[m_2],\ldots,E[m_\ell])$.

Given $E[b]$ and $E[b']$ ($b,b'\in\{0,1\}$), FHE allows to compute the bitwise operations $E[b\oplus b']$ and 
$E[b\cdot b']$ without decrypting $E[b]$ and $E[b']$.
It also allows all arithmetic and logical operations via XOR and AND as follows:
Let $x,y$ represent $\ell$-bit integers and $x_i,y_i$ the $i$-th bit of $x,y$, respectively.
Let $c_i$ represent the $i$-th carry-in bit and $s_i$ the $i$-th bit of the sum $x+y$.
Then, we can obtain the private full-adder $E[x+y]$ using the relations
$s_i = x_i\oplus y_i\oplus c_i$ and $c_{i+1}=(x_i\oplus c_i)\cdot (y_i\oplus c_i)\oplus c_i$.
We can construct other operations such as subtraction, multiplication, and division based on the full-adder.

On the other hand, we can also obtain the private comparison $E[cmp(>,x,y)]$ satifying 
$cmp(>,x,y) = 1$ if $x>y$ and $0$ otherwise, because $cmp(>,x,y)$ is identical to
the most significant bit of $y+(-x)$, which can be obtained without decryption.
Here, $(-x)$ is the bit complement of $x$ obtained by $x_i\oplus 1$ for all $i$-th bits.

Thus, by using FHE, it is possible to perform all arithmetic operations and comparisons for encrypted variables or elements of arrays. 
For simplicity, we represent such operations in plaintext notation throughout the remainder of this paper. 
That is, unless otherwise stated or unless it may cause confusion, an operation such as the addition of two encrypted integers $E[x]$ and $E[y]$, 
resulting in $E[x + y]$, will be denoted simply as $x + y$.

Based on FHE, we can design a private sorting algorithm as follows (e.g.~\cite{Bonnoron2017}).
Note that, in this code, the variables {\tt gt}, {\tt tmp}, and any entry in the array {\tt arr[1..n]} are all encrypted.
Therefore, after this code is executed, the resulting {\tt arr[1..n]} is sorted securely
without revealing the private information about {\tt arr[1..n]}.

\begin{verbatim}
// private bubble_sort over ciphertexts
void bubble_sort(int *arr,int n){
    for(int i=0;i<n-1;i++){
        for(int j=1;j<n-1;j++){
            int gt=cmp(>,arr[j-1],arr[j]); 
            int tmp=gt*arr[j-1]^(!gt)*arr[j]; 
            arr[j-1]=gt*arr[j]^(!gt)*arr[j-1];
            arr[j]=tmp;
        }
    }
}
\end{verbatim}

However, it is difficult to improve the $O(n^2)$ time algorithm to $O(n\log n)$. This is because encrypting pointers is fundamentally meaningless. If pointers are in plaintext, information about the input (e.g., the distribution of values) may be leaked. On the other hand, if the pointers are encrypted, it becomes impossible to access the corresponding memory addresses, and thus computation cannot proceed. Satisfying these conflicting requirements simultaneously appears to be impossible.
This difficulty can be partially avoided by using {\em sorting network}, where comparisons are performed only between fixed pairs of elements, eliminating the need for pointers. Theoretically, we can construct an optimum sorting network that achieves $O(n \log n)$ comparisons~\cite{Ajtai1983}, but the circuit size is impractically large. Therefore, in practice, $O(n \log^2 n)$ algorithm~\cite{Batcher1968} is used. This study also adopts this approach.
\section{Method}
\subsection{Sorting Network on FHE}
Sorting via FHE is typically performed using sorting networks. It remains an open question whether a fastest algorithm such as quicksort or mergesort can be effectively realized in FHE setting. In this study, we adopt the Batcher's odd-even mergesort~\cite{Batcher1968} as the basis for sorting.
Batcher's algorithm employs a recursively constructed sorting network according to the number of elements and achieves sorting in $O(n \log^2 n)$ time. However, it requires that the number of data items be a power of two. Therefore, if the input data size does not meet this condition, we insert appropriate dummy elements beforehand to adjust the size.
Moreover, since sorting network is basically not a stable sort, we attach a $\log n$-bit suffix to each data element to ensure stability. This preprocessing does not affect the correctness or quality of feature selection in any way.

\subsection{A Naive Algorithm for Private CWC}
In this study, we assume that the data analyst (i.e., an algorithm) receives the input triple $(D,F,C)$ in encrypted form, and that the algorithm is capable of performing FHE operations on encrypted data. 

Here, $(D,F,C)$ is possessed as an array, and the algorithm has access to any encrypted entry $d_i(f_j)$ or $d_i(C)$ for $i\in\{1,2,\ldots,n\}$ and $j\in\{1,2,\ldots,k\}$.
The private CWC (denoted pCWC) utilizes sorting $D$ with the feature vectors $\vec{F}=(f_1,f_2,\ldots,f_k)$ formed by $f_i \in F$. 
By sorting $(D,F,C)$ with $\vec{F}$ as the key, identical feature vectors are placed adjacent to each other, significantly speeding up the consistency checking step.

Table~\ref{feature-label} (upper) shows an input $(D,F,C)$ 
and Table~\ref{feature-label} (lower) shows the sorted $(D,F,C)$ where
each original index $d_i$ is renamed by the sorted order $D_j$
and for each $i\in\{1,2,\ldots,k\}$, the corresponding feature label $L[i][j]$ is 
the unique grouping label of $\log n$ bits for $D_j$ defined by the feature vector $\vec{F}[..i]=(f_1, f_2, \ldots, f_i)$.  

For example, Table~\ref{feature-label} (lower), the labels $L_5$ for $D_1$ and $D_2$ are both $000$, indicating that $D_1$ and $D_2$ belong to the same group 
based on identical feature vector values up to $f_5$.
Therefore, the consistency check needs to be performed only for data points that share the same label associated with the tail $f_i$.  
Thus, by using the sorted $(D, F, C)$ along with the corresponding feature labels in $L[1..k][1..n]$, the consistency checking process becomes significantly simplified.

\begin{table}[H]
\caption{
Sorted $(D,F,C)$ and corresponding feature labels $L_i$.
} 
\label{feature-label}
\newcolumntype{C}{>{\centering\arraybackslash}X}
\begin{tabularx}{\textwidth}{CCCCCCC}
\hline
\boldmath{$D$} & $f_1$ & $f_2$ & $f_3$ & $f_4$ & $f_5$ & \;\boldmath{$C$}\; \\ \hline
$d_1$& 1& 0 &1& 0& 0& 1 \\
$d_2$& 0 & 1 & 0 & 0 &  1 & 0 \\
$d_3$& 0 & 1 &0&0 &1& 0 \\
$d_4$& 1 &0& 0& 0& 1& 1 \\
$d_5$& 1& 0&1& 1 &1& 0 \\ 
\hline
\hline
\boldmath{sorted $D$} & $f_1: L[1][1..5]$ & $f_2: L[2][1..5]$ & $f_3: L[3][1..5]$ & $f_4: L[4][1..5]$ & $f_5: L[5][1..5]$ & \;\boldmath{$C$}\; \\ \hline
$D_1(=d_2)$& 0:000 & 1:000 & 0:000 & 0:000 &  1:000 & 0 \\
$D_2(=d_3)$& 0:000 & 1:000 &0:000 &0:000 &1:000& 0 \\
$D_3(=d_4)$& 1:001 &0:001& 0:001& 0:001& 1:001& 1 \\
$D_4(=d_1)$& 1:001& 0:001 &1:010& 0:010& 0:010& 1 \\
$D_5(=d_5)$& 1:001& 0:001&1:010& 1:011 &1:011& 0 \\ 
\hline
\end{tabularx}
\end{table}

Algorithm~\ref{naive-pcwc} shows a naive pCWC computation based on the sorted $(D, F, C)$ and the feature labels.  
Given $(D, F, C)$, consider the evaluation of each feature $f_t$ in $F = \{f_1, f_2, \ldots, f_k\}$.
The algorithm evaluates the consistency on the feature set $F' = F \setminus \{f_t\}$, obtained by removing $f_t$, 
and obtains the corresponding Boolean value $b_t$.  
Specifically, $b_t$ is computed as:
\[
b_t = \bigwedge_{i=2}^n \left( (D_{i-1} \neq D_i) \lor (D_{i-1}(C) = D_i(C)) \right)
\]
Here, $b_t = 1$ iff $f_t$ is selected.  
By reporting the resulting $\vec{b}=(b_1,\ldots,b_k)$ to the data owner, the selected features are obtained
after decrypting $\vec{b}$.

We next analyze the complexity of our algorithm where 
in the following discussion, we assume that the cost of each homomorphic operation under FHE is constant. 
That is, the computational cost of arithmetic and comparison operations on encrypted data is treated as $O(1)$, unless otherwise specified.

\begin{algorithm}[H]   
\caption{Naive pCWC for ciphertexts}                         
\label{naive-pcwc}                          
\begin{algorithmic}[1] 
\REQUIRE An encrypted $(D,F,C)$;
\ENSURE A minimal consistent subset of $F$;
\FOR{$t=k,k-1,\ldots,1$} 
\STATE sort $(D,F',C)$ with $\vec{F'}$ as the key for $F'=F\setminus \{f_t\}$;
\STATE compute $L[j][i]$ for all $j=1,2,\ldots,k$ and $i=1,2,\ldots,n$;
\STATE compute the consistency $b_t\in\{0,1\}$ of $(D,F',C)$;
\STATE update $D_i(f_t)\leftarrow D_i(f_t)\cdot b_t$ for all $i=1,2,\ldots,n$;
\ENDFOR
\RETURN $\vec{b}=(b_1,b_2,\ldots,b_k)$ representing the selected subset of $F$;
\end{algorithmic}
\end{algorithm}

\begin{theorem}~\label{Th1}
Algorithm~\ref{naive-pcwc} (naive pCWC) on ciphertexts simulates Algorithm~\ref{cwc} (CWC) on plaintexts.
The time and space complexities are $O(k^2n\log^3 n)$ and $O(kn)$ for $|F|=k$ and $|D|=n$, respectively.
\end{theorem}
\begin{proof}
Assuming that $(D, F', C)$ is already sorted, the feature label $L[t][i]$ of $D_i$ for $\vec{F}[..t]$ can be computed based on 
whether $D_{i-1}(..f_t) = D_i(..f_t)$ holds, where $D_i(..f_t)$ stands for $D_i(f_1)\cdots D_i(f_t)$.
When the predecessor $L[t][i-1]$ for $D_{i-1}$ is already defined, using the logical bit $x$ to indicate 
whether $D_{i-1}(..f_t) = D_i(..f_t)$, the next $L[t][i]$ is defined as:
$L[i] \leftarrow L[i-1] + \neg{x}$.
These computations can be performed on FHE.
Moreover, the logical bit $b_t$, which determines whether the feature $f_t$ should be selected, can also be computed by FHE according to its definition.  
By updating each $f_t$ as: $D_i(f_t) \leftarrow D_i(f_t) \cdot b_t$,
we can effectively remove $f_t$ from $(D, F, C)$.
Thus, Algorithm~\ref{naive-pcwc} correctly simulates the original CWC on plaintexts.

When sorting the encrypted $(D, F', C)$ using a sorting network, the number of comparisons required is $O(kn \log^2 n)$.  
Since our method appends a $\lceil \log n\rceil$-bit suffix to each data entry to achieve stable sorting, 
the computation time for one feature $f_t$ becomes $O(kn \log^3 n)$.  
Therefore, the total computation time across all features is $O(k^2 n \log^3 n)$.
The space complexity is clear.
\end{proof}

\subsection{Improved Private CWC}
In the naive approach, sorting of $(D, F', C)$ is iterated for each $F'=F \setminus \{f_t\}$ with specified $f_t\in F$,
resulting in high computational cost. Therefore, we aim to accelerate this algorithm. 
It should be noted that, since $(D, F, C)$ are encrypted, pointer-based referencing cannot be used. 
In this study, we propose a method that avoids full sorting at each step by partitioning $F$ 
into two sequences---selected features and unprocessed features---and processing them independently before merging the results.

The key concept here is to utilize the feature labels in $L[t][1..n]$ introduced in Algorithm~\ref{naive-pcwc} for partial sorting. 
Suppose that we are focusing on a current feature $f_t$. 
That is, feature selection has already been completed for the suffix vector $\vec{F}[t+1..]=(f_{t+1}, f_t,\ldots, f_k)$. 
In this case, we define the feature labels for $\vec{F}[t+1..]$ and store them in the array $PostL[1..n]$ as follows: 
$PostL[i] = PostL[i']$ iff $D_i(f_{j+1}..) = D_{i'}(f_{j+1}..)$. 
This groups the suffix sequences of $F$ that share the same feature vector. 

As shown in Algorithm~\ref{naive-pcwc}, feature labels for any prefix vector $\vec{F}[..t]$ are stored in the array $L[t][1..n]$. 
Therefore, to group $(D, F \setminus \{f_t\}, C)$, we first independently compute $L[t][1..n]$ for prefix vector $\vec{F}[..t]$
and $PostL[1..n]$ for the corresponding suffix vector $\vec{F}[t+1..]$, 
and then calculate the ranks of $L[t][1..n]$ and $PostL[1..n]$ to merge them correctly. 
This approach enables us to avoid iterative sorting for the entire $(D, F', C)$.

\begin{algorithm}[H]   
\caption{Improved pCWC for ciphertexts}                         
\label{improved-pcwc}                          
\begin{algorithmic}[1] 
\REQUIRE An encrypted $(D,F,C)$;
\ENSURE A minimal consistent subset of $F$;
\STATE sort $(D,F,C)$ with $\vec{F}$ as the key;
\STATE compute the feature label $L[j][i]$ for all $j=1,2,\ldots,k$ and $i=1,2,\ldots,n$;
\STATE initialize $PostL[1..n]\leftarrow (0,0,\ldots,0)$ and $MapL[1..n]\leftarrow (1,2,\ldots,n)$;
\FOR{$t=k,\ldots,1$} 
\STATE sort $(f_{t+1}[1..n], PostL[1..n]$, $C[1..n]$, $MapL[1..n])$ as the key $f_{t+1}[1..n]$;
\STATE update $PostL[1..n]$ by the sorted $f_{t+1}[1..n]$;
\STATE sort $(f_t, L[t-1][1..n])$ by the inverse $MapL^{-1}[1..n]$;
\STATE compute the consistency $b_t\in\{0,1\}$ of $(L[t-1][1..n],PostL[1..n],C[1..n])$;
\STATE update $f_t: D_\ell(f_t)\leftarrow D_\ell(f_t)\cdot b_t$ for all $\ell=1,2,\ldots,n$;
\ENDFOR
\RETURN $\vec{b}=(b_1,b_2,\ldots,b_k)$ representing the selected subset of $F$;
\end{algorithmic}
\end{algorithm}

The task of Algorithm~\ref{improved-pcwc} is divided into three phases: preprocessing on $(D, F, C)$ (Figure~\ref{pcwc-1}), 
partial sorting of $(D, F, C)$ with respect to a specific feature $f_t$ (Figure~\ref{pcwc-2}), 
and the decision-making and update of $(D, F, C)$ based on $f_t$ (Figure~\ref{pcwc-3}), respectively.

Figure~\ref{pcwc-1} illustrates the preprocessing of $(D, F, C)$. Here, it is assumed that the sorting of $(D, F, C)$ has already been completed, 
and each $D_i$ represents its rank in the sorted order. For this sorted $(D, F, C)$, a label $L_t$ is assigned to each $D_i$ based on $\vec{F}[..t]$, 
such that $L[t][i] = L[t][j]$ iff $D_i(..f_t) = D_j(..f_t)$. 
That is, $L[t][1..n]$ group labels for $D_i$s sharing the same value of $\vec{F}[..t]$.

Next, we focus on a specific feature $f_t$ and determine whether it should be selected. 
For this purpose, we require the sorting result of $(D, F', C)$, where $F' = F \setminus \{f_t\}$ (see Figure~\ref{pcwc-2}). 
However, since $\vec{F}[..t-1]$ has already been sorted, we divide $F$ into $\vec{F}[..t-1]$, $f_t$, and
$\vec{F}[t+1..]$ with the current $f_t$, and only $\vec{F}[t+1..]$ is sorted .
Because the sorting of $\vec{F}[..t-1]$ is performed sequentially for $t = k, k-1, \ldots, t$, it only needs to be executed for the current value of $f_t$. 
As a result, the computational cost is significantly reduced compared to naively sorting both the prefix and suffix for every possible $f_t$.

However, performing sorting on $\vec{F}[..t-1]$ breaks the alignment between the previously synchronized $\vec{F}[..t-1]$ and $\vec{F}[t+1..]$. 
To keep this correspondence, we introduce a reference array $MapL[1..n]$ initialized by $(1,2,\ldots,n)$, and sort $(PostL[1..n],MapL[1..n])$ simultaneously. 
We then compute the inverse array $MapL^{-1}[1..n]$ by sorting $(1, 2, …, n)$ with $MapL[1..n]$ as the key,
and by sorting $MapL^{-1}[1..n]$ together with $\vec{F}[t+1..]$, 
the synchronization between the prefix and corresponding suffix of $\vec{F}$ is maintained.
Through this procedure, we obtain the correctly sorted result of $(D, F', C)$ without $f_t$ excluded.

Figure~\ref{pcwc-3} illustrates the task for verifying the consistency of the sorted $(D, F', C)$ and determines whether $f_t$ is a necessary feature. 
Since this information is extracted as an encrypted logical bit $b_t$, the decision can be reflected back to $(D, F, C)$ by computing the product of $b_t$ with all values of $f_t$.
These operations are carried out sequentially in the order $t = k, k-1, \ldots, 1$, and by decrypting the resulting bit vector $\vec{b}=(b_1, b_2, \ldots, b_k)$, 
we obtain the final feature selection result.

\begin{figure}[H]
\begin{center}
\includegraphics[width=14cm]{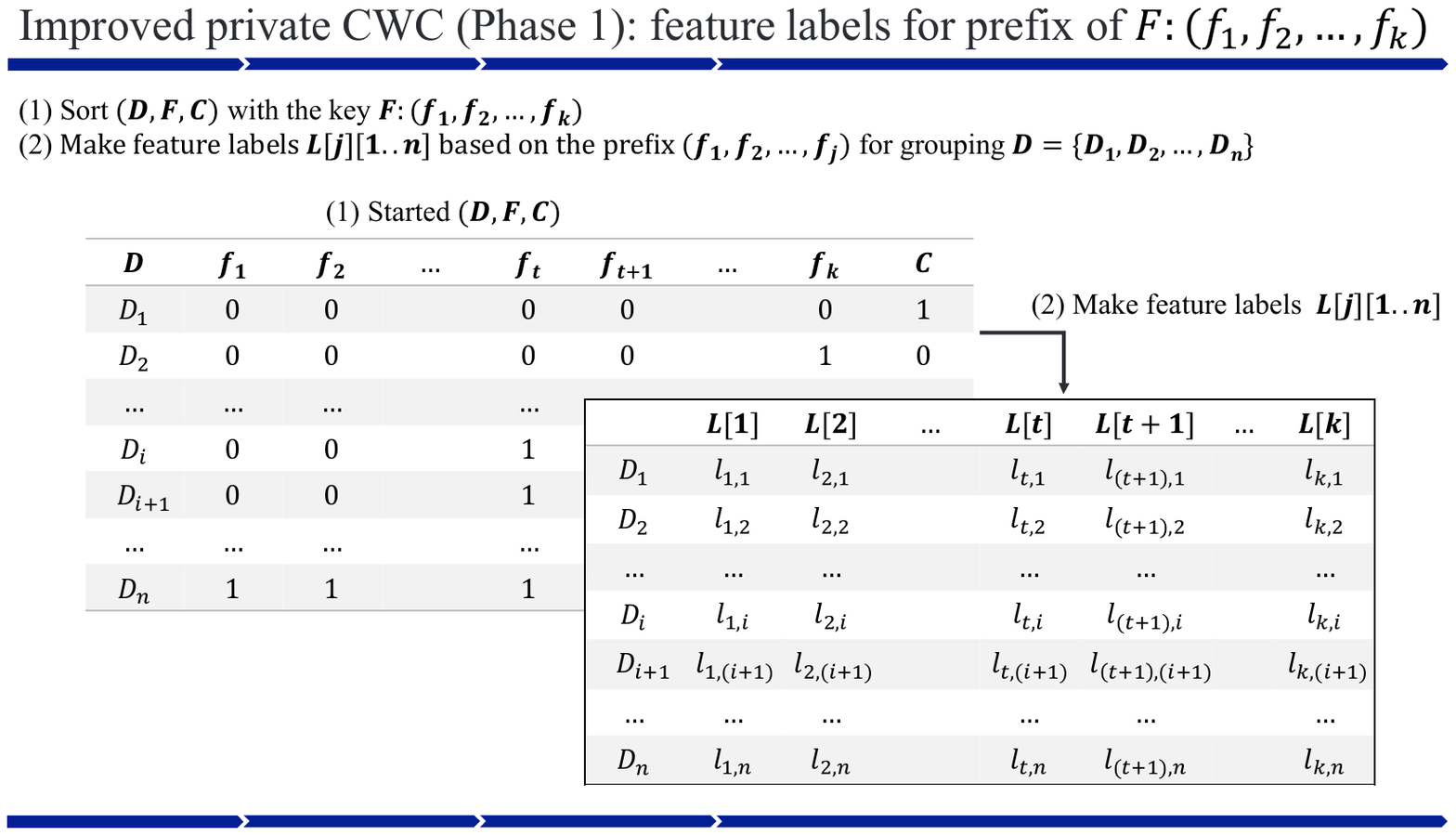}
\end{center}
\vspace{-1cm}
\caption{Phase 1 of the improved algorithm corresponding to Line 1 - 3 in Algorithm~\ref{improved-pcwc}.~\label{pcwc-1}}
\end{figure}   

\begin{figure}[H]
\begin{center}
\includegraphics[width=14cm]{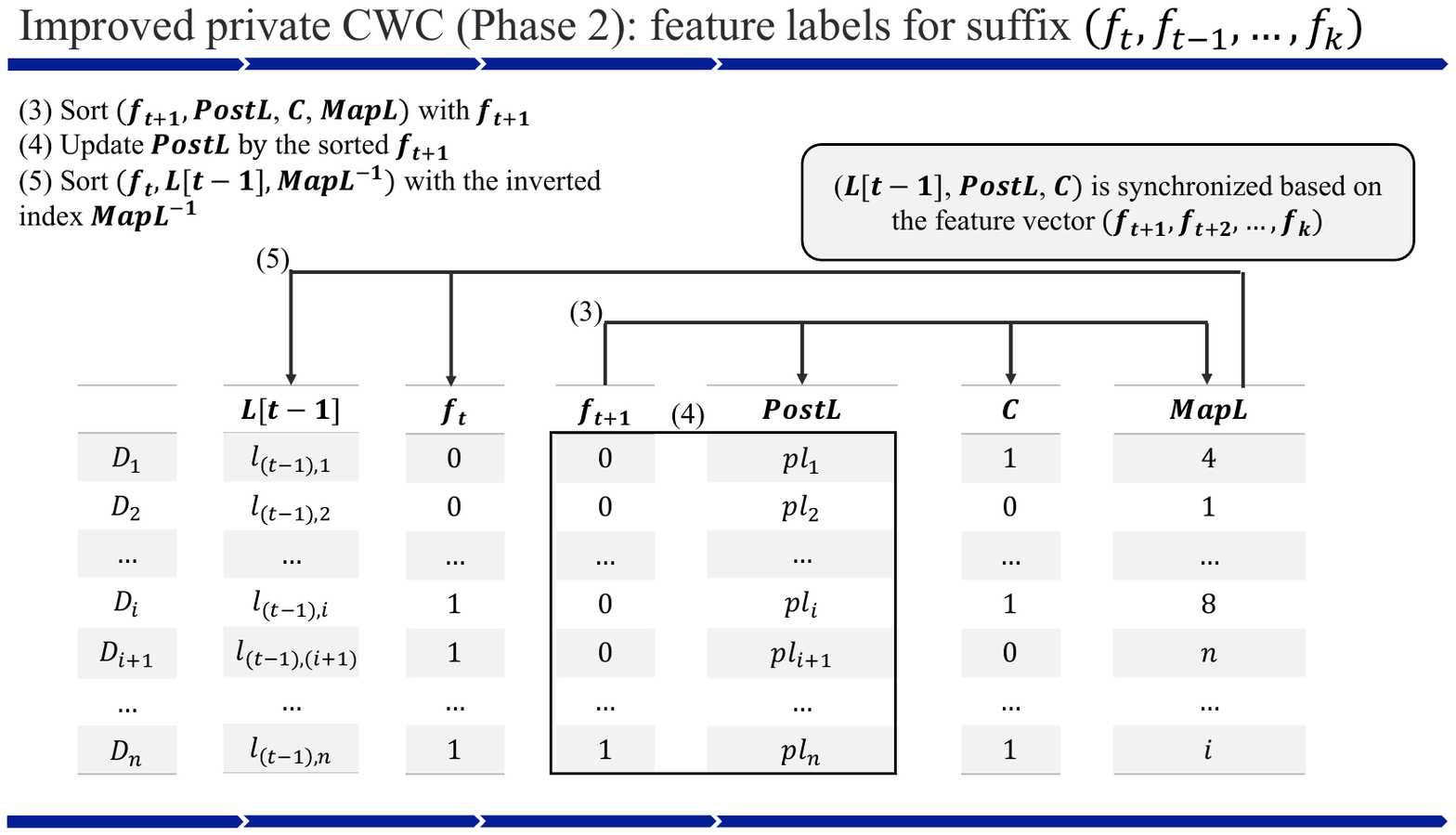}
\end{center}
\vspace{-1cm}
\caption{Phase 2 (Line 4 -7) of Algorithm~\ref{improved-pcwc}.~\label{pcwc-2}}
\end{figure}   

\begin{figure}[H]
\begin{center}
\includegraphics[width=14cm]{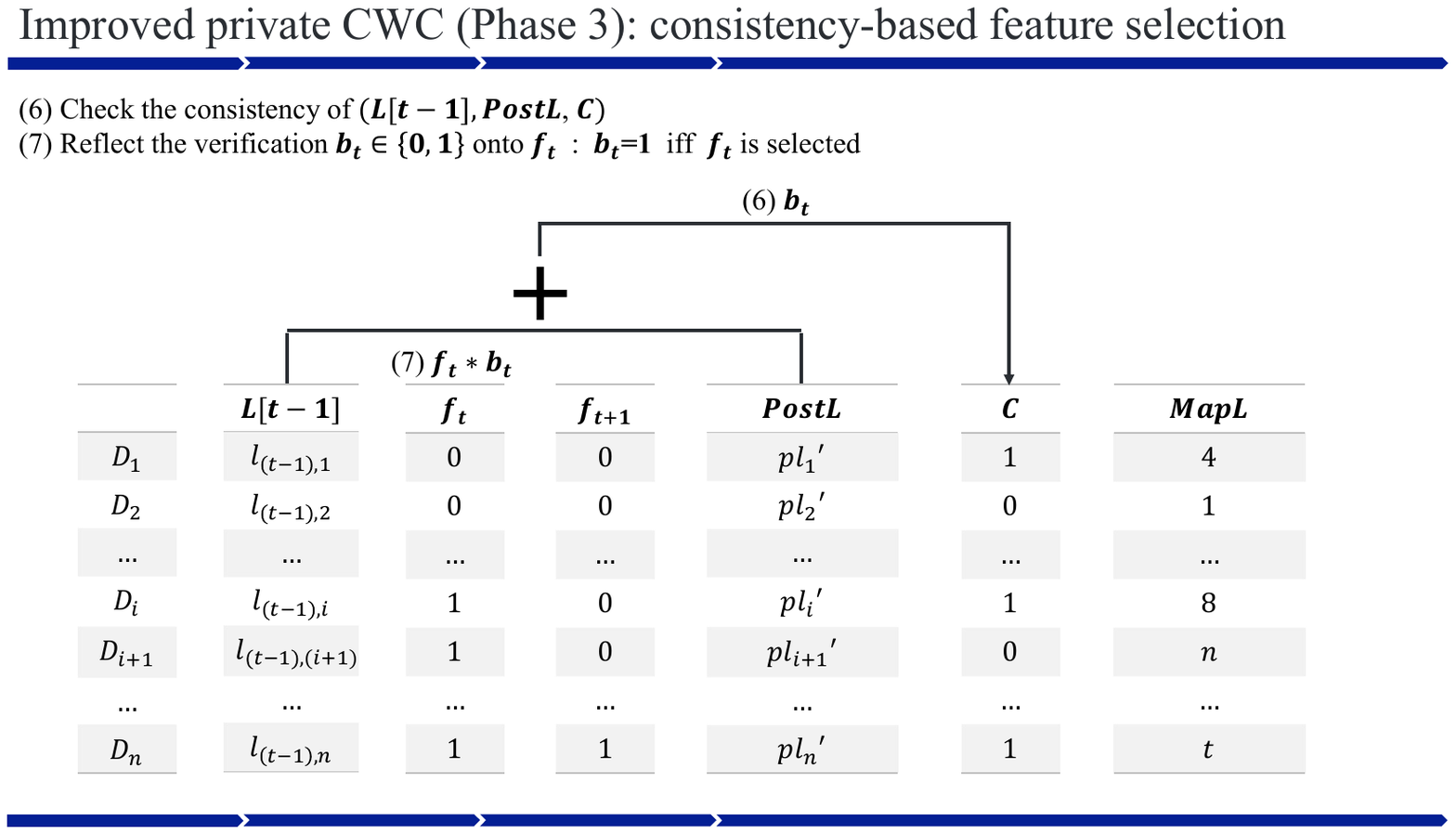}
\end{center}
\vspace{-1cm}
\caption{Phase 3 (Line 8 - 9) of Algorithm~\ref{improved-pcwc}.~\label{pcwc-3}}
\end{figure}   

\begin{theorem}~\label{Th2}
Algorithm~\ref{improved-pcwc} (improved pCWC) on ciphertexts simulates Algorithm~\ref{cwc} (CWC) on plaintexts.
The time and space complexities are $O(kn\log^3 n)$ and $O(kn)$ for $|F|=k$ and $|D|=n$, respectively.
\end{theorem}
\begin{proof}
We first show the correctness of Algorithm~\ref{improved-pcwc}, that is it correctly simulates Algorithm~\ref{naive-pcwc}.  
To establish this, it is sufficient to prove that for any feature $f_t$, Algorithm~\ref{improved-pcwc} 
(1) correctly sorts $(D, F \setminus \{f_t\}, C)$, and (2) accurately determines the consistency of the sorted $(D, F \setminus \{f_t\}, C)$.

As a result of preprocessing, for any $t$, $(D, F, C)$ is already sorted using the prefix vector $\vec{F}[..t]=(f_1, f_2,\ldots, f_t)$ as the key, 
and thus the correct label $L[t][i]$ for $D_i$ has been computed.

Now, assuming that the feature labels $PostL[1..n]$ corresponding to the suffix $\vec{F}[t+2..]$ has been computed in the previous loop, 
we can obtain the updated ranks for the extended suffix $\vec{F}[t+1..]$ by sorting $PostL[1..n]$ using $f_{t+1}$ as the key. 
Then, by updating $PostL[1..n]$ based on the sorted values and their corresponding $f_{t+1}$ values, 
we can compute the new $PostL[1..n]$ that reflects the extended suffix $\vec{F}[t+1..]$.
All of these computations can be performed under FHE, in the same manner as in Algorithm~\ref{naive-pcwc}.

The next necessary step is to restore the correspondence between $L[t-1][1..n]$ and $PostL[1..n]$. 
The array $L[t-1][1..n]$ has retained its original order from the initial sorting for prefixes, whereas the current $PostL[1..n]$ 
has been reordered according to the current suffix.
Suppose that the previous $PostL[1..n]$ corresponding to the suffix $\vec{F}[t+2..]$ was already synchronized with $L[t-1][1..n]$.

To preserve the previous ordering of $PostL[1..n]$ (i.e., the ordering synchronized with $L[t-1][1..n]$), 
we first sort the initialized index array $MapL[1..n] = (1, 2, \ldots, n)$ together with $PostL[1..n]$. 
Then, since the original positions of $MapL[1..n]$ are retained in the inverse index array $MapL^{-1}[1..n]$, 
we can restore the correct correspondence between $L[t-1][1..n]$ and $PostL[1..n]$ by sorting $L[t-1][1..n]$ using $MapL^{-1}[1..n]$ as the key.
All of these operations can be performed solely through sorting on encrypted data under FHE.

Moreover, since both $f_{t+1}$ and $C$ are always kept in synchronization with $PostL[1..n]$, and each $f_t$ is always sorted in accordance with 
$L[t-1][1..n]$, Algorithm~\ref{improved-pcwc} correctly sorts $(D, F \setminus \{f_t\}, C)$. 
The consistency check applied to the sorted $(D, F \setminus \{f_t\}, C)$ is identical to that in Algorithm~\ref{naive-pcwc}. 
Therefore, we conclude that the improved algorithm correctly simulates the naive one.

Next, we evaluate the computational complexity.  
The sorting of $(D, F, C)$ in Lines~2--4 is identical to that of Algorithm~\ref{naive-pcwc}, and thus requires $O(kn \log^3 n)$ time.  

Among Lines~6--15, the most computationally expensive operations are the sorting steps in Line~8 and Line~10.  
In Line~8, sorting is performed using $f_{t+1}$ as the key. To ensure stability, a $\log n$-bit suffix is appended to each entry, 
allowing the operation to be treated as sorting integers of $O(\log n)$ bits. Therefore, the time complexity is $O(n \log^3 n)$.

In Line~10, sorting is performed using $MapL[1..n]$ as the key. Since $MapL[1..n]$ is initialized as $(1, 2, \ldots, n)$, 
no additional suffix is required for stable sorting, and this operation also runs in $O(n \log^3 n)$ time.
Because the above process is repeated $k$ times, the overall computational time is $O(kn \log^3 n)$.

Finally, the additional data structures used by Algorithm~\ref{improved-pcwc}, besides $(D, F, C)$, 
include $L[1..k][1..n]$, $PostL[1..n]$, $MapL[1..n]$, and $MapL^{-1}[1..n]$.  
Since each of these has a size of $O(kn)$, the total space complexity of the algorithm is also $O(kn)$.
\end{proof}
\section{Experimental Results}
We implemented the proposed private CWC (Algorithm~\ref{improved-pcwc}) in C++. 
For FHE operations, we used TFHE library~\cite{TFHE}.  
The experiments were conducted on a machine equipped with an Intel(R) Core(TM) i9-10900X CPU running at 3.70GHz and 32~GB of memory.  
We used \texttt{gcc} version 11.4.0 as the compiler.  
All experiments were executed inside a Docker container specifically built for this purpose.

To evaluate the effectiveness of the proposed method (Algorithm~\ref{improved-pcwc}), we compared it with the naive baseline approach (Algorithm~\ref{naive-pcwc}). 
In the baseline algorithm, label computation becomes a bottleneck due to the need for $k$ sorting operations, which makes the method impractical for large values of $k$. 
To address this issue, a decision-tree-based data structure can be used to classify encrypted feature vectors 
and assign labels without sorting~\cite{Bost2015,Paul2022}. 
Assuming this data structure, the computational complexity of the baseline algorithm becomes $O(kn^2 \log n)$, 
and then, the total complexity of Algorithm~\ref{naive-pcwc} can be considered as $O(\min\{ k^2 n \log^3 n, kn^2 \log n \})$.  
In contrast, the improved method (Algorithm~\ref{improved-pcwc}) has a complexity of $O(kn \log^3 n)$.  
We experimentally compared the execution times of these algorithms while varying $k$ and $n$.

\begin{figure}[H]~\label{time}
\centering
\begin{minipage}[b]{0.45\textwidth}
\centering
\includegraphics[width=\textwidth]{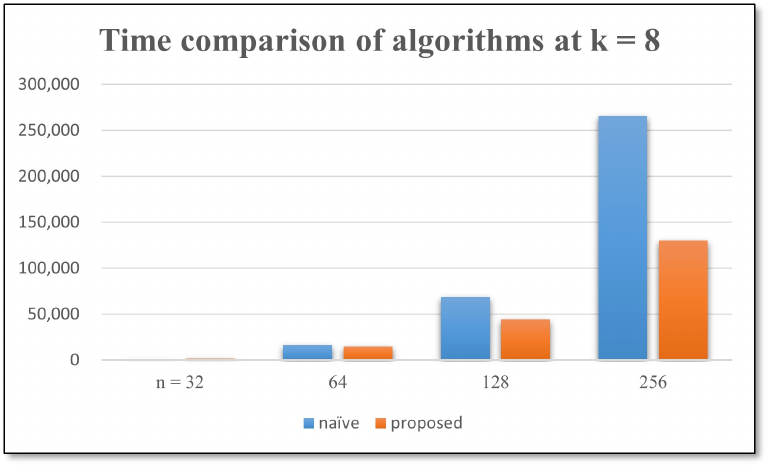}
\caption{Time (sec.) comparison of proposed and naive algorithms}~\label{time-1}
\end{minipage}
\hfill
\begin{minipage}[b]{0.45\textwidth}
\centering
\includegraphics[width=\textwidth]{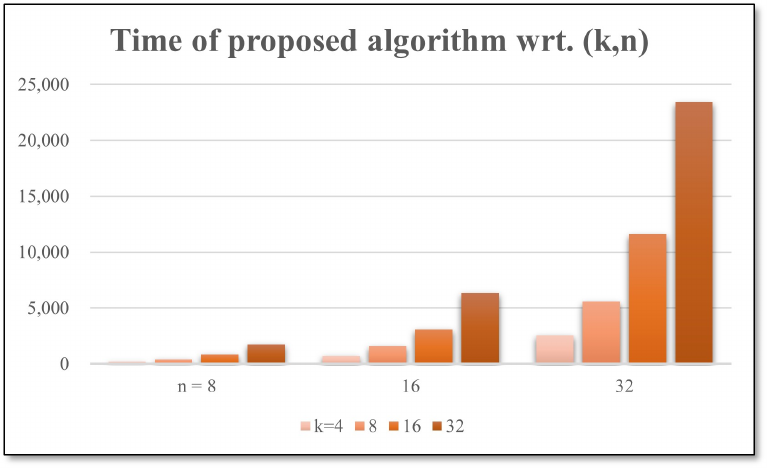}
\caption{Time (sec.) of proposed algorithm w.r.t. parameters}~\label{time-2}
\end{minipage}
\end{figure}   

Figure~\ref{time-1} compares the computation time of the proposed algorithm, which has a theoretical complexity of $O(kn \log^3 n)$, 
with a naive implementation that uses decision-tree-based secure computation~\cite{Bost2015,Alabdulkarim2019,Paul2022} and has a complexity of $O(kn^2 \log n)$.
In this experiment, we measured the computation time of both algorithms as the number of records $n$ increases.  
The results clearly show that the proposed method reduces computation time even for relatively small values of $n$, 
and this advantage becomes more pronounced as $n$ increases.

On the other hand, Figure~\ref{time-2} shows the impact of the parameters $(k, n)$ on the computation time of the proposed algorithm.  
In this experiment, we compared the execution time for $k \in \{4, 8, 16, 32\}$ and $n \in \{8, 16, 32\}$.  
Although the theoretical time complexity of the proposed method is $O(kn \log^3 n)$, 
the results empirically confirm that the computation time increases linearly with $k$.
Therefore, it is confirmed that the proposed privacy-preserving computation algorithm behaves as expected by design.
\section{Discussion}
In this study, we proposed a feature selection algorithm based on FHE.  
While various plaintext feature selection methods have been studied, 
consistency-based feature selection has been shown to offer both scalability and effectiveness. 
Therefore, the development of privacy-preserving computation protocols for consistency-based feature selection is a critical research challenge.

Although existing FHE-based two-party algorithms have addressed this goal to some extent, the protocol proposed in this work is the first to realize 
{\em fully outsourced} computation for consistency-based feature selection.  
We demonstrated that our proposed privacy-preserving algorithm is superior both theoretically and experimentally compared to 
a naive approach that uses decision-tree-based data structures under secure computation.

On the other hand, the proposed algorithm has several directions for future improvement.  
First, this work assumes that the input data $(D, F, C)$ is binary. 
Removing this assumption and extending the algorithm to support multi-valued or symbolic attribute domains is one of the most important future enhancements.  
Since TFHE library~\cite{TFHE} used in this study, is optimized for bitwise and integer operations, 
alternative libraries or methods capable of handling real-valued data more efficiently should be considered.
In addition, while this study assumes a point-to-point computation model between data owner and analyst, 
extending the model to support one-to-many computations is also important for broader applicability.  
Such an extension could be realized by leveraging multi-key FHE~\cite{Ma2021,MKHE_CCS2019,Cheon2025,Namazi2025}.

By addressing these aspects, the proposed algorithm can be extended and adapted to enable privacy-preserving feature selection 
across a wide range of application domains such as FHE-based machine learning~\cite{Yuan2025,Naresh2025,Babu2025,Kolhar2023,Xiao2019} 
and federated machine learning~\cite{Firdaus2025,Zhang2025,Zhao2023,Walskaar2023}.



\end{document}